\documentclass[12pt]{amsart}
\usepackage[english]{babel}
\usepackage{amsmath,amsthm,amsfonts,amssymb,epsfig,vector,color,ulem}
\usepackage[left=1in,top=1in,right=1in]{geometry}

\newtheorem{thm}{Theorem}[section]
\newtheorem{lem}[thm]{Lemma}
\newtheorem{prop}[thm]{Proposition}
\newtheorem{cor}[thm]{Corollary}
\newtheorem{df}[thm]{Definition}

\newtheorem{ass}[thm]{Assumption}

\newcommand{\E}{\mathbb{E}}

\newcommand{\prob}{\mathbb{P}}

\newcommand{\R}{\mathbb{R}}
\newcommand{\Z}{\mathbb{Z}}
\newcommand{\N}{\mathbb{N}}
\newcommand{\F}{\mathcal{F}}

\begin{document}

\title[Fluctuations of interface free energies in spin glasses]{Fluctuation Bounds For Interface Free Energies\\ in Spin Glasses}

\author[L.-P. Arguin]{L.-P. Arguin}            
 \address{L.-P. Arguin\\ 
 D\'epartement de Math\'ematiques et Statistique\\
 Universit\'{e} de Montr\'{e}al\\
Montr\'{e}al, QC, Canada}
\thanks{L.-P. A. is supported by a NSERC discovery grant and a grant FQRNT {\it Nouveaux chercheurs}.}
\email{arguinlp@dms.umontreal.ca}

\author[C.M. Newman]{C.M. Newman}            
 \address{C.M. Newman\\ 
 Department of Mathematics\\
 University of California\\
 Irvine, CA 92697; NYU-Shanghai and Courant Institute of Mathematical Sciences\\
 New York University\\
 New York, NY 10012, USA
 }
\thanks{The research of CMN and DLS is supported in part by U.S.~NSF Grants DMS-1207678 and OISE-0730136.}
\email{newman@cims.nyu.edu}

\author[D.L.~Stein]{D.L.~Stein}            
 \address{D.L.~Stein\\ 
 Department\ of Physics and Courant Institute of Mathematical Sciences\\
  New York University\\
 	 New York, NY 10003, USA
 }
\thanks{}
\email{daniel.stein@nyu.edu}

\author[J. Wehr]{J. Wehr}            
 \address{J. Wehr\\ 
Department of Mathematics\\
University of Arizona\\
Tucson, AZ 85721, USA
 }
\thanks{}
\email{wehr@math.arizona.edu}


\date{April 8, 2014}



\maketitle

\begin{abstract}
  We consider the free energy difference restricted to a finite volume
  for certain pairs of incongruent thermodynamic states  (if they exist) in
  the Edwards-Anderson Ising spin glass at nonzero temperature.  We
  prove that the variance of this quantity with
  respect to the couplings grows proportionally to the volume
  in any dimension greater than or equal to two. As an illustration
  of potential applications, we use this result to restrict the
  possible structure of Gibbs states in two dimensions.
\end{abstract}

\section{Introduction}
\label{sec:intro}

The quantitative dependence of free energy fluctuations on the
quenched disorder is a fundamental thermodynamic property of all
disordered systems. 
Bounds on such fluctuations were used in~\cite{AW90}
to prove uniqueness of the Gibbs state
in a class of two-dimensional disordered systems, including the
two-dimensional random~field Ising~model. 
In the case of a
finite-volume spin glass with a specified boundary condition, the free
energy is a random function of the couplings within the volume: as the
couplings vary, so does the free energy. One may also consider for a
fixed volume the {\it difference\/} in free energy between two
different boundary conditions, say periodic and antiperiodic; this
quantity is similarly a random function of the couplings. Because such
a change in boundary conditions induces an interface between their
respective finite-volume Gibbs states, one can refer equivalently to
the dependence of the resulting {\it interface free energy\/} on the
couplings.  The interface between two spin configurations (in the
present case, each chosen from a different Gibbs state corresponding
to the two different boundary conditions) is defined as the set of
couplings satisfied in one configuration and unsatisfied in the other.

It has been noted by several
authors~\cite{ENS06,HF87,FH87,FH88,NS92,NSjpc03,NSbook} that
determining quantitative bounds on the interface free energy variance
with respect to the couplings would allow a resolution of a
long-standing open problem at the heart of the statistical mechanics
of short-range spin glasses in finite dimensions: how many pure states
(or at zero temperature, ground states) are present in the spin glass
phase?  To begin, one has to specify what one means by multiple
states; this was elucidated nicely in~\cite{HF87} where it was
noted that in the context of spin glasses it is particularly important
to differentiate between {\it incongruent\/} and {\it regionally
  congruent\/} pure states. In the case of $d$-dimensional Ising spin
glasses, the latter are related by a global spin flip everywhere
except near interfaces whose dimensionality is strictly smaller than
$d$; otherwise, the states are incongruent. So, for example,
incongruent {\it ground\/} states differ by a ``space-filling''
interface. 
Whether incongruent pure and ground states exist
has been a subject of controversy for many years. One of the main goals in studying free energy fluctuations of
the corresponding interfaces is to provide the tools for settling this question.

More generally, for a Gibbs state $\Gamma$, consider a correlation
function $\Gamma(\sigma_{x_0}\sigma_{y_0})$ (often denoted
$\langle\sigma_{x_0}\sigma_{y_0}\rangle_\Gamma$) for two fixed,
nearest-neighbor sites $x_0$ and $y_0$. Then multiple incongruent or
regionally congruent states are present if there exists more than one
thermodynamic limit for such a correlation function (along different
subsequences of volumes), say $\Gamma(\sigma_{x_0}\sigma_{y_0})$ and
$\Gamma'(\sigma_{x_0}\sigma_{y_0})$~\cite{FH87}. Incongruent states
are those in which such~$(x_0,y_0)$ constitute a nonvanishing fraction
of edges:
\begin{df}
\label{df: incongruent} 
Two states $\Gamma$ and $\Gamma'$ are incongruent if for some
$\varepsilon>0$ there is a subset of edges $(x_0,y_0)$ with strictly
positive density such that $|\Gamma(\sigma_{x_0}\sigma_{y_0}) -
\Gamma'(\sigma_{x_0}\sigma_{y_0})|>\varepsilon$.
\end{df}
It is known~\cite{NS01a} that if one uses coupling-{\it independent\/}
boundary conditions (which include the usual ones of free, periodic,
or fixed, chosen along a deterministic sequence of volumes) to
generate distinct states, then these states must be incongruent.
Because of this natural property, they were referred to
in~\cite{NSjpc03} as {\it observable\/} states.

In contrast, regionally congruent states, which differ by zero-density
interfaces such as those found in homogeneous ferromagnets using
Dobrushin boundary conditions, could also exist in spin glasses in
principle, but if so could only be generated using coupling-{\it
  dependent\/} boundary conditions~\cite{NS01a}, using procedures as
yet unknown.  Such states, if they exist, would be of mathematical
interest, but whether their presence would carry any physical or
thermodynamical significance is at present unclear. In what follows we
will concern ourselves exclusively with incongruent pure or ground
states.

Several authors, using various techniques, have considered the problem
of free energy fluctuations of a finite-volume Ising spin glass with a
single boundary condition~\cite{AM03,Chatterjee09,Contucci12,ENS06} In
contrast, there are few results for the variance of free energy
differences. Based on heuristic scaling arguments, it was proposed
in~\cite{FH86} that, in general dimension, 
 the variance of the fluctuations of such free energy differences is bounded, up to a constant, by the system's surface area.
This conjecture was
proved for free energy differences between
finite-volume Gibbs states generated by gauge-related boundary
conditions such as periodic and antiperiodic \cite{AFunpub,NSunpub} (see also Sect.~5.7 in~\cite{CG13}).  A key quantity that
remains to be determined is therefore a {\it lower\/} bound for
interface free energy fluctuations.  Indeed, the various competing
pictures of the low-temperature spin glass phase predict different
size dependences of such fluctuations. Determining a lower bound,
however, has remained elusive due to a number of problems
which will be discussed below. In this paper we provide an initial
step by obtaining such a lower bound, for the free energy difference
within a finite volume for certain classes of distinct infinite-volume
Gibbs states  (cf. Assumption \ref{ass}) assuming such classes exist.  We note that the lower bound obtained here is for a
quantity that differs somewhat from that for which an upper bound has
been obtained,
which prevents us from applying
of our result to dimensions higher than two.

\section{Basic Outline of Proof}
\label{sec:outline}

Before proceeding with the technical details of the proof, we present
a brief sketch of its main ideas and basic ingredients. We will always
work in the periodic boundary condition
metastate~\cite{AW90,NS96,NS97,NS98}, although our results can be
extended to any metastate generated using coupling-independent
boundary conditions (free, fixed, and so on).  The technical
definition of the metastate will be presented in
Sect.~\ref{sec:prelim}, but for now a rough definition will suffice:
the metastate can be thought of as a probability measure on all of the
thermodynamic states that appear in an infinite sequence of volumes
with specified boundary conditions.

One longstanding technical hurdle in studying free energy fluctuations
for spin glass Hamiltonians is the problem of ``keeping track'' of the
selected~$\Gamma$ and~$\Gamma'$ among the set of all possible Gibbs
states as couplings vary. In contrast to Hamiltonians for ordered
systems (such as ferromagnets), there is no known connection between
spin glass Gibbs~states and boundary conditions --- or more precisely,
none that can be made in a translation-invariant, measurable way
(which will be needed later). There are several possible ways to
address this issue. For example, if at least some of the Gibbs~states
have nonzero weight in the metastate, one can use these weights as
state ``identifiers''.  A different way to solve the above issue is to
randomly sample $\Gamma$ and $\Gamma'$ among the Gibbs~states in the
metastate(s). We will examine both of these possibilities in this
paper.

A second hurdle concerns the ``cancellation problem'': unlike the
ferromagnet (including the random-field Ising model), the sign of the
energy difference between two spin configurations varies with the
couplings that constitute the interface. Consequently, although each
edge will presumably contribute an energy fluctuation of order one,
the fluctuations may cancel as one moves along the interface. This
could in principle lead to an overall energy fluctuation that varies
sublinearly with the volume, or may even be volume-independent. The
latter may indeed be the case {\it a~priori\/}, and indeed corresponds
to some proposed scenarios of the spin glass phase. The problem is
that, because of these cancellations, any of the usual techniques used
to study fluctuations of free energy differences between incongruent
states could lead to poor estimates of their volume dependence and
therefore inconclusive results.

In order to get more accurate estimates of the volume dependence of
the free energy difference, one needs to find a suitable quantity to
study. There are in principle a number of different possible choices,
all of which behave like the free energy difference between two
states, and all of which should be physically indistinguishable; but
most are not suitable for obtaining a lower bound. More specifically,
previous attempts to obtain a good lower bound failed because the
choice of a free energy difference-type quantity --- typically,
looking at the free energy difference in a finite volume with two
different boundary conditions --- interfered with the overall
translation-covariance properties of the spin glass model in all
space.  Hence, one key feature of the argument in this paper relies on
the choice of a quantity which, although it also is a finite-volume
free energy difference, manages to preserve the overall
translation-invariance properties of the infinite system. There are
several possible choices; we will use a quantity similar to that used
in~\cite{AW90} (see in particular Eq.~(5.11)), which provides several
technical advantages to alternative possibilities.

Preservation of translation-invariance is necessary but not sufficient
to circumvent the cancellation problem. The presence of incongruence,
along with the suitable choice of a quantity to measure free energy
differences between states, allows one to condition free energy
differences on individual couplings in a translation-invariant way;
this results in the desired random variable (for now, call it
$\mu_{xy}$) which reflects the extent of free energy difference
fluctuations on the coupling $J_{xy}$, but whose distribution does not
depend on the choice of bond $(x,y)$.  The most straightforward way to
arrive at such a quantity is to average over all couplings except one,
but (in a manner similar to that described in the preceding paragraph)
such averaging can lead to cancellations such that the variance of
$\mu_{xy}$ tends to zero as the volume increases. 

To address this issue, we adopt a modified approach, using blocks
rather than individual bonds. That is, we now consider $\mu_{xy}$ to
be conditioned on a {\it block\/} $b$ centered at $(x,y)$. Because
$\mu_{xy}$ as a function of the couplings is not a constant, the
variance of $\mu_{xy}$ will not tend to zero for $b$ sufficiently
large; this follows from the martingale convergence theorem, as
discussed in Sect.~\ref{sec: lower}.  Importantly, the size of $b$ is
independent of both $(x,y)$ and $\Lambda$, because of the translation
covariance of $\mu_{xy}$. Consequently, in this approach different
blocks of the large finite volume under study contribute equally to
the variance, leading to a stronger lower bound on the growth of
fluctuations, which indeed are found to grow as the square root of the
volume.

With these considerations in mind, we now sketch the main steps of the
proof. As noted, we work in the periodic boundary condition metastate,
which has strong coupling- and translation-covariance properties built
in.  These covariance properties, along with the
translation-invariance of the coupling distribution itself, naturally
extend to measures defined on the couplings and/or the thermodynamic
states on which the metastate is supported. These ``nice'' properties
allow one to equate a (suitably defined) derivative of the free energy
difference, with respect to a specific coupling $J_{xy}$, to
$\Gamma(\sigma_x\sigma_y)-\Gamma'(\sigma_x\sigma_y)$, the latter being
the difference in the indicated two-spin correlation function
evaluated in $\Gamma$ and $\Gamma'$, respectively. The assumption of
incongruent states implies this difference will be of order~one with
positive probability in the couplings.

The next step is to construct a {\it martingale
  decomposition\/}~\cite{Durrett91} of the free energy difference.
This involves dividing the volume $\Lambda$ into blocks whose size is
independent of $\Lambda$, but which are large enough so that the
metastate average of
$\Gamma(\sigma_x\sigma_y)-\Gamma'(\sigma_x\sigma_y)$, conditioned on
the couplings in the block, is nonzero. We then consider a sequence of
(metastate averages of) the free energy differences in $\Lambda$, with
each element in the sequence conditioned on the couplings in an
increasing number of blocks. The variance of the free energy
difference over {\it all\/} the couplings in $\Lambda$ is easily shown
to be no smaller than the sum of the variances of the differences
between the $k^{\rm th}$ and the $(k-1)^{\rm th}$~elements of the
sequence, with $k$ running from 1~to~$N$, the number of blocks in the
full volume.

The remainder of the proof uses the translation-covariant properties
of the metastate to show that the variance of the $k^{\rm
  th}$~difference is independent of $k$. Because (by the assumption of
existence of incongruence) the variance of the first block is of
order~one independent of the volume (this result doesn't precisely
follow from the incongruence assumption, but requires only a small
amount of further work), the result follows.

In this paper we do not consider upper bounds of the same quantity in
dimensions greater than two.  In two dimensions, however, some results
can be obtained by following the approach of~\cite{AW90} (see also
\cite{bovier} for a detailed exposition).  They are based on the
elementary fact that the free energy difference (in any dimension) is
bounded uniformly by the sum of the couplings on the boundary of
$\Lambda$; in two dimensions, one can arrive at a contradiction using
the martingale central limit theorem, but following~\cite{AW90}, it is
sufficient to show that the values  of
the free energy difference rescaled by $|\Lambda|^{1/2}$ are
unbounded.  While the results are not surprising --- one expects in
fact a single pure thermodynamic state in two dimensions at all
positive temperatures --- they illustrate the potential power of the
techniques described below to determine the structure of the low
temperature spin glass phase.

\section{Preliminaries}
\label{sec:prelim}

We consider the Edwards-Anderson~(EA) Hamiltonian~\cite{EA75} on a finite box $\Lambda=[-L,L]^d\subset \Z^d$
\begin{equation}
\label{eq: EA}
H_{\Lambda, J}(\sigma)=-\sum_{(x,y)\in E(\Lambda)}J_{xy} \sigma_x\sigma_y, \qquad \sigma\in \{-1,1\}^{\Lambda} \ .
\end{equation}
For $B\subset \Z^d$, we write $E(B)$ for the set of edges with both
ends in $B$.  The couplings $J_\Lambda:=(J_{xy}, (x,y)\in E(\Lambda))$
are i.i.d.~$\nu$-sampled random variables.  (We will sometimes abuse
notation and write $\nu$ for the joint distribution of the couplings
as well as its marginals.)  We assume throughout that $\nu (dJ_{xy})$
is continuous and that $\int \nu(dJ_{xy}) J_{xy}^4<\infty$.

To proceed we need to introduce some notation. Write $\Sigma=\{-1,+1\}^{\Z^d}$ and let $\mathcal M_1(\Sigma)$ be the set of (regular Borel) probability measures on $\Sigma$. 
An infinite-volume Gibbs state~$\Gamma$ for the Hamiltonian~\eqref{eq: EA} is an element of $\mathcal M_1(\Sigma)$ that satisfies the DLR equations~\cite{DLR} for that Hamiltonian. 
The Gibbs state~$\Gamma$ induces a probability measure on infinite-volume spin configurations: for a function $f(\sigma)$ on the spins
\begin{equation}
\label{eq:Gamma}
\Gamma(f(\sigma))=\int\ d\Gamma\  f(\sigma)\, .
\end{equation}
The set of Gibbs states corresponding to the coupling realization~$J$ is denoted by $\mathcal G_J$.
The inverse temperature $\beta$ is fixed throughout.

We are now in a position to define the quantity we will be studying.
In a finite box $\Lambda$ and at inverse temperature
$0\le\beta<\infty$, consider the {\it difference of free energies}
between two infinite-volume Gibbs states $\Gamma$ and $\Gamma'$ for
the Hamiltonian~\eqref{eq: EA}
\begin{equation}
\label{eq: F}
F_\Lambda(J,\Gamma,\Gamma')=\log \frac{\Gamma(\exp \beta H_{\Lambda,J}(\sigma))}{\Gamma'(\exp \beta H_{\Lambda,J}(\sigma'))}\ .
\end{equation}
Note that $\Gamma(\exp \beta H_{\Lambda,J}(\sigma))$ is a ratio whose
numerator comprises a Boltzmann factor on spin configurations only in
$\Lambda^c$, the {\it complement\/} of $\Lambda$, while the
denominator is the usual Boltzmann factor on all infinite-volume spin
configurations. Consequently, $F_\Lambda(J,\Gamma,\Gamma')$ amounts to
the difference in free energies within the volume
$\Lambda$ between boundary conditions chosen from the infinite-volume thermodynamic states $\Gamma'$ and $\Gamma$.

For the coupling realization $J$, consider now two probability measures on $\mathcal M_1(\Sigma)$ that are supported on the set of Gibbs states; we will denote them $\kappa_J$ and $\kappa'_J$.
We will study the free energy fluctuations of two independently chosen states $\Gamma$ and $\Gamma'$ under the probability measure
\begin{equation}
\label{eq: M}
M:=\nu(dJ) \ \kappa_J (d\Gamma')\times \kappa'_J(d\Gamma)\ .
\end{equation}
(We will also consider below a case where $\Gamma$ and $\Gamma'$ are
chosen from the same $\kappa_J$.)  The most natural measures
$\kappa_J$ to consider are the ones that are obtained by taking
(deterministic) subsequential limits of finite-volume Gibbs measures,
as discussed in~\cite{AW90,AD11,NS96,NS97,NS98}.  In so doing the
$\kappa_J$'s will inherit two useful invariance properties from the
finite-volume Gibbs measures.

The first of these is translation-covariance. Formally, let $T$ be a
translation of $\Z^d$, and consider the operation mapping $\mathcal
M_1(\Sigma)$ to itself: $\Gamma \mapsto T\Gamma$, where
\begin{equation}
\label{eq:transcov}
T\Gamma(f(\sigma)):= \Gamma(f(T\sigma))\ .
\end{equation}

The second  is covariance under a local modification of the couplings:
for $B\subset \Z^d$ finite and $J_B\in \R^{E(B)}$, we define the operation
$L_{J_B} : \Gamma \mapsto L_{J_B}\Gamma$ where
\begin{equation}
\label{eq: gamma L}
\Big(L_{J_B} \Gamma\Big)(f(\sigma))= \frac{\Gamma\Big(f(\sigma) \exp\Bigl(-\beta H_{B,J}(\sigma)\Bigr)\Big)}{\Gamma\Big(\exp\Bigl(-\beta H_{B,J}(\sigma)\Bigr)\Big)}\, .
\end{equation}
This simply modifies the couplings within a finite subset $B$ of
$\Z^d$. It was shown in~\cite{AW90,AD11,NS96,NS97,NS98} that the
$\kappa_J$'s arising as subsequence limits of finite-volume Gibbs
measures are probability measures on Gibbs states. Such measures are
referred to as {\it metastates\/}~\cite{NS96}.
\begin{df}
\label{df: metastate}
A metastate $\kappa_{\cdot}$ for the EA Hamiltonian on $\Z^d$ is a measurable mapping 
\begin{equation}
\begin{aligned}
\R^{E(\Z^d)} &\to \mathcal M_1(\Sigma)\\
J &\mapsto \kappa_{J}
\end{aligned}
\end{equation}
with the following properties
\begin{enumerate}
\item {\bf Support on Gibbs states.} Every state sampled from the $\kappa_J$ is a Gibbs state for the realization for the couplings. Precisely,
\begin{equation}
\kappa_{J}\Bigl(\mathcal G_{J}\Bigr)=1 \ \text{$\nu$-a.s.}
\end{equation}
\item {\bf Coupling Covariance. } For  $B\subset \Z^d$ finite, $J_B\in \R^{E(B)}$, and any measurable subset $A$ of  $\mathcal M_1(\Sigma)$,
\begin{equation}
\kappa_{J+J_B}(A)= \kappa_{J}(L_{J_B}^{-1}A)
\end{equation}
where $L_{J_B}^{-1} A=\Bigl\{\Gamma\in \mathcal M_1(\Sigma): L_{J_B}\Gamma \in A\Bigr\}$.
\item {\bf Translation Covariance.} For any translation $T$ of $\Z^d$ and any measurable subset $A$ of  $\mathcal M_1(\Sigma)$
\begin{equation}
\kappa_{TJ}(A)=\kappa_{J}(T^{-1}A)\ .
\end{equation}
\end{enumerate}
\end{df}
The translation covariance is a direct consequence when finite-volume Gibbs measures with periodic boundary conditions are considered.
For other coupling-independent boundary conditions, such as free or periodic, it can be recovered by taking an average of the translates of 
the finite-volume Gibbs measures~\cite{AW90,ADNS10,NS01b}. 

\section{Main results}
\label{sec:main}

Our goal is to understand how the random variable $F_\Lambda$ fluctuates under the measure $M$ defined in~\eqref{eq: M}.
Of course, the appearance of nontrivial fluctuations requires that the states $\Gamma$ and $\Gamma'$ differ with positive probability:
\begin{equation}
\label{ass: ideal}
M\Bigl\{(J,\Gamma,\Gamma'): \Gamma\neq \Gamma'\Bigr\}>0\ .
\end{equation}
Recall that two elements of $\mathcal M_1(\Sigma)$ differ if and only
if there exists a correlation function for which the expectations in
$\Gamma$ and $\Gamma'$ are different.  In this paper, we look at the
simplest case of~\eqref{ass: ideal} where the first moment of the
one-edge correlation function differs.  This implies the existence of
incongruent Gibbs states.  Our working assumption will therefore be:
\begin{ass}
\label{ass}
There exists an edge $(x,y)\in E(\Z^d)$ such that 
\begin{equation}
\label{eq: ass}
\nu\Bigl\{ J: \kappa_J\big(\Gamma(\sigma_x\sigma_y)\big)\neq \kappa'_J\big(\Gamma'(\sigma'_x\sigma'_y)\big)\Bigr\}>0\ .
\end{equation}
\end{ass}

Note that if the assumption is satisfied for one edge then it is satisfied for all edges, by the translation covariance of the metastates and the translation invariance of the coupling distribution $\nu$.
Also, the assumption implies the existence of an incongruent pair $(\Gamma,\Gamma')$ in the sense of Definition \ref{df: incongruent}:
 there must exist $\varepsilon>0$ for which
\begin{equation}
\lim_{\Lambda\to\Z^d} \frac{1}{|E(\Lambda)|} \sum_{(x,y)\in E(\Lambda)}\kappa_J\times\kappa_J'
\left(
 |\Gamma(\sigma_x\sigma_y)- \Gamma'(\sigma'_x\sigma'_y)|
\right)>\varepsilon \ \text{on a set of $J$'s of positive probability.}
\end{equation}
Thus, $\kappa_J\times\kappa'_J$ must sample incongruent pairs.
We remark that the existence of incongruent pairs does not imply Assumption \ref{ass} since the quantity $\Gamma(\sigma_x\sigma_y)- \Gamma'(\sigma'_x\sigma'_y)$ might fluctuate yet have zero mean when sampled from $\kappa_J\times \kappa_J'$.

Our main result is to prove a lower bound
on the fluctuations of $F_\Lambda$ of the order of square root of the
volume in any dimension under this incongruence assumption.
\begin{thm}
\label{thm: lower}
If Assumption \ref{ass} holds, then there exists a constant $c>0$ such that for any $\Lambda=[-L,L]^d \subset \Z^d$
the variance of $F_\Lambda$ under $M$ satisfies
\begin{equation}
\text{Var}_M\Big(F_\Lambda\Big)\geq c |\Lambda |\ .
\end{equation}
\end{thm}
The theorem is proved in Section~\ref{sec: lower}, using a martingale decomposition of $F_\Lambda$ as discussed in Sec.~\ref{sec:outline}. 

\bigskip

{\it Remark.\/} The results of~\cite{WA90} (see in particular
Theorem~2.3) imply that the variance of the free energy difference
satisfies an {\it upper\/} bound of the same form but with a larger
constant, as long as the coupling distribution has a finite second
moment. 

\bigskip

Lower bounds of the type of Theorem~\ref{thm: lower} can be used to rule out certain structures of the metastates and of the underlying Gibbs state in dimension $d=2$.
As an example, we prove the following corollary in Sect.~\ref{sec: d=2}.
\begin{cor}
\label{cor: 2d}
In $d=2$, for every pair of metastates $\kappa_\cdot$ and $\kappa'_\cdot$ and every edge $(x,y)\in E(\Z^d)$, 
\begin{equation}
\label{eq:2d1}
\kappa_J(\Gamma(\sigma_x\sigma_y))=\kappa'_J(\Gamma'(\sigma'_x\sigma'_y)) \ \text{ $\nu$-a.s.}
\end{equation}
\end{cor}

It is important to remark that Assumption~\ref{ass} can be used to
study the structure of a single metastate.  If we suppose the
metastate is supported on more than one state, then it is possible to
condition this measure on two different subsets of incongruent states
in order to produce two different measures $\kappa_J$ and
$\kappa_J'$. Then we would be in position to use Theorem \ref{thm:
  lower} and, in two dimensions, Corollary \ref{cor: 2d}.  However,
the subsets on which we are conditioning must be chosen carefully so
that the coupling covariance and the translation covariance of a
metastate still hold for the conditional measures.  To illustrate this
method, we apply it to rule out possible structures of the metastate
in dimension two.
\begin{cor}
\label{cor: countable}
A metastate $\kappa_\cdot$ in $d=2$ cannot be supported on a countably infinite number of incongruent states, that is it cannot be of the form
$$
\kappa_J(d\Gamma)=\sum_{\alpha\in \mathcal A}p_\alpha\delta_{\Gamma^\alpha}
$$
for a countable index set $\mathcal A$, $\sum_\alpha p_\alpha=1$ and
with incongruent states $\Gamma^\alpha\in \mathcal G_{J}$ in the sense
of Definition~\ref{df: incongruent}
\end{cor}
The important consequence of the countability hypothesis is the
existence of two distinct weights $p_\alpha\neq p_{\alpha'}$, which is
ensured in the countably infinite case.  The proof applies also to the
case where $\mathcal A$ is finite and the weights are not all equal.
Of course, the result is far from the conjecture that in $d=2$ there
exists a single Gibbs state.  Our hope is that this method of proof
using fluctuation bounds could be extended to further restrict the
possible structures of the set of Gibbs states in dimension two and
higher.

\section{Proof of Theorem \ref{thm: lower}}
\label{sec: lower}

We start by introducing some notation and reviewing some elementary facts.

For a finite $B\subset \Z^d$, we write $J_B:=(J_{xy}, (x,y)\in E(B))$
for the couplings on $E(B)$.  We write $\partial B$ for the set of
edges that have only one end in $B$.  We denote by $J_{B^c}$ the set
of couplings with at least one end outside $B$; that is,
$J_{B^c}=(J_{xy}, (x,y) \in E(B^c)\cup \partial B)$.  For a metastate
$\kappa_\cdot$ and a finite set $B$, we write $\kappa_{J_{B^c}}$ for
the metastate conditioned on the realization $J$ but with $J_B$ set to
zero.

For two metastates $\kappa_\cdot$ and $\kappa'_\cdot$, and a finite set $B\subset \Z^d$, we define the measure
\begin{equation}
\label{eq: M comp}
M_{B^c}:= \nu(dJ_{B^c}) \ \kappa_{J_{B^c}} \times \kappa'_{J_{B^c}}\ .
\end{equation}
We shall consider the measure $\nu(dJ_B) M_{B^c}$. 
The following Lemma proves some important invariance properties of the measure~$M$ that will be needed later.


\begin{lem}[Invariance properties of $M$]
\label{lem: M invariance}
Let $\kappa_\cdot$ and $\kappa'_\cdot$ be two metastates and $B$ a finite box in $\Z^d$. 
Consider the measures $M$ and $M_{B^c}$ given in \eqref{eq: M} and in \eqref{eq: M comp}. Then for any measurable function $F: \R^{E(\Z^d)}\times \mathcal M_1(\Sigma)\times \mathcal M_1(\Sigma)\to \R$ we have both
\begin{equation}
\text{$F(J,\Gamma,\Gamma')$ has the same law under $M$ as under $TM$ for any translation $T$}
\end{equation}
and
\begin{equation}
\text{ $F(J, L_{J_B}\Gamma, L_{J_{B}}\Gamma') $ has the same law under $\nu(dJ_B) M_{B^c}$ as $F(J,\Gamma,\Gamma')$ under $M$\ .}
\end{equation}
\end{lem}

\bigskip

\begin{proof}
The coupling covariance of the metastates implies that, for any Borel~measurable function  $F: \R^{E(\Z^d)}\times \mathcal M_1(\Sigma)\times \mathcal M_1(\Sigma)\to \R$,
$F(J, L_{J_B}\Gamma, L_{J_{B}}\Gamma') $ has the same law under $\nu(dJ_B) M_{B^c}$ as $F(J,\Gamma,\Gamma')$ under $M$.
Furthermore, it is easily checked that translation covariance of the metastate, together with the fact that the distribution on the couplings is itself invariant under translation, gives translation-invariance of the measure: if $T$ is a translation in $\Z^d$ and $F: \R^{E(\Z^d)}\times \mathcal M_1(\Sigma)\times \mathcal M_1(\Sigma)\to \R$ is some measurable function, then
$TM(F):=M(T^{-1}F)=M(F)$. Here, $M(F)=\int F\ dM$ is the expectation of $F$ under $M$.  Note that by translation invariance of $M$, if Assumption \ref{ass} holds for one edge, it holds for all edges. 
\end{proof}

\bigskip

The next lemma states, without a proof, two elementary properties of the variance that will be needed.
\begin{lem}
For $X$ a random variable on $(\Omega,\F,\prob)$ and a sub $\sigma$-algebra $\mathcal G\subset \mathcal F$, 
\begin{equation}
\label{eq: var}
\text{ Var}X = \E[\text{Var}(X| \mathcal G)] + \text{Var}(\E[X|\mathcal G])\ .
\end{equation}
Also, if $X$ and $X'$ are two independent copies of $X$, then
\begin{equation}
\label{eq: sym var}
\text{Var}X= \frac{1}{2} \E[(X-X')^2]\ .
\end{equation}
\end{lem}

The crucial ingredients of the proof of Theorem \ref{thm: lower} are the next two lemmas. 
Recalling the definition of $M_{B^c}$ in~\eqref{eq: M comp}:
\begin{lem}
\label{lem: cond}
Consider a measurable function $F: \R^{E(\Z^d)}\times \mathcal M_1(\Sigma)\times \mathcal M_1(\Sigma)\to \R$ and a finite subset $B$ of $\Z^d$.
Then
\begin{equation}
M\big( F(J,\Gamma,\Gamma') | J_B\big)= M_{B^c}\big(F(J, L_{J_B} \Gamma,  L_{J_B} \Gamma'\big)\ \text{ $\nu$-a.s.}
\end{equation}
where $M(\cdot|J_B)$ is the conditional expectation given $J_B$. 
\end{lem}
\begin{proof}
Let $g(J_B)$ be some bounded measurable function of $J_B$. We need to show that
\begin{equation}
M\big(F(J,\Gamma, \Gamma') g(J_B)\big)= \int \nu(dJ_B) \ M_{B^c}\big(F(J, L_{J_B} \Gamma,  L_{J_B} \Gamma')\big) g(J_B)\ .
\end{equation}
But  by Lemma~\ref{lem: M invariance} the left side equals 
\begin{equation}
\int \nu(dJ_B)\int \nu(dJ_{B^c}) \ \kappa_{J_{B^c}}\times \kappa_{J_{B^c}}\big( F(J, L_{J_B} \Gamma,  L_{J_B}  \Gamma') g(J_B)\big)
\end{equation}
which is, by definition and by taking $g(J_B)$ inside the first expectation,
\begin{equation}
\int \nu(dJ_B) \ g(J_B) M_{B^c}\big(F(J, L_{J_B} \Gamma,  L_{J_B} \Gamma') \big) \ ,
\end{equation}
as claimed.
\end{proof} 
One consequence of Lemma~\ref{lem: cond} is to provide a smooth version of the conditional expectation of the free energy difference given $J_B$. 
In particular, we can compute its derivatives.
\begin{lem}
\label{lem: deriv}
Let $F_\Lambda$ be as in \eqref{eq: F} and $B\subset \Lambda$ finite. Then for any edge $(x,y)$ in $B$, we have
\begin{equation}
\frac{\partial}{\partial J_{xy}} M\big(F_\Lambda | J_B\big)= \beta \ M_{B^c} \big( L_{J_B}\Gamma(\sigma_x\sigma_y)-  L_{J_B}\Gamma'(\sigma_x\sigma_y)\big) \text{ $\nu$-a.s.}
\end{equation}
\end{lem}
\begin{proof}
By Lemma \ref{lem: cond}, we need to compute
\begin{equation}
\frac{\partial}{\partial J_{xy}}M_{B^c}\Bigl(F(J, L_{J_B} \Gamma,  L_{J_B} \Gamma)'\Bigr)
= 
\frac{\partial}{\partial J_{xy}}M_{B^c}\Biggl( \log  L_{J_B} \Gamma\Bigl(\exp\beta H_{J,\Lambda}(\sigma)\Bigr) -\log  L_{J_B} \Gamma'\Bigl(\exp\beta H_{J,\Lambda}(\sigma')\Bigr)  \Biggr)\ .
\end{equation}
By \eqref{eq: gamma L}, the first  logarithmic term on the right side is simply
\begin{equation}
\log  L_{J_B} \Gamma(\exp\beta H_{J,\Lambda}(\sigma))= \log \frac{\Gamma\big(\exp\beta H_{J,\Lambda}(\sigma)\exp -\beta H_{J,B}(\sigma)\big)}{\Gamma\big(\exp -\beta H_{J,B}(\sigma)\big)}
\end{equation}
Note that the numerator no longer depends on $J_B$ since its contribution is cancelled. Therefore,
\begin{equation}
\frac{\partial}{\partial J_{xy}} \log  L_{J_B} \Gamma(\exp\beta H_{J,\Lambda}(\sigma))= 
\beta \frac{\Gamma\big(\sigma_x\sigma_y \ \exp -\beta H_{J,B}(\sigma)\big)}{\Gamma\big(\exp -\beta H_{J,B}(\sigma)\big)}= L_{J_B}\Gamma(\sigma_x\sigma_y)\ .
\end{equation}
It remains to prove that the derivative can be passed through $M_{B^c}$. This follows from dominated convergence by noticing that the derivatives are uniformly bounded, completing the proof.

\end{proof}

\medskip

With these lemmas in hand, we can now proceed to the proof of the main result of the paper.

\begin{proof}[Proof of Theorem~\ref{thm: lower}]  Let $\F_\Lambda=\sigma\bigl(J_{x,y}, (x,y) \in E(\Lambda)\bigr)$.  By \eqref{eq: var}, we have
\begin{equation}
\text{Var}_M( F_\Lambda) \geq  \text{Var}_M \Bigl( M(F_\Lambda| \F_\Lambda)\Bigr)\, .
\end{equation}
Next divide $\Lambda$ in equally sized blocks $B_1,\cdots, B_N$, with $1\ll |B_k|\ll|\Lambda|$. The size of the blocks is independent of $\Lambda$, so that
\begin{equation}
N= C|\Lambda|
\end{equation}
for some constant $C>0$.
The size of each $B$ is sufficiently large so that
\begin{equation}
\label{eq: key}
M\Bigl(\Gamma(\sigma_x\sigma_y)- \Gamma'(\sigma_x\sigma_y)  | J_B\Bigr)\neq 0 \text{ with positive $\nu$-probability.}
\end{equation}
Indeed, we get by taking $B\to \Z^d$ that
\begin{equation}
M\Bigl(\Gamma(\sigma_x\sigma_y)- \Gamma'(\sigma_x\sigma_y)  | J_B\Bigr)\to \kappa_J \Bigl(\Gamma(\sigma_x\sigma_y)\Bigr)- \kappa'_J \Bigl(\Gamma(\sigma_x\sigma_y)\Bigr)\, ,
\end{equation}
where we used the Martingale Convergence Theorem for uniformly integrable martingales (see, for example, Theorem 5.5.7 in \cite{Durrett91}).
Thus, there would be a contradiction with Assumption~\ref{ass} if there were a sequence of volumes $B\to \Z^d$ such that \eqref{eq: key} is zero 
almost surely along the sequence.

We now consider $\F_k=\sigma(J_{B_i}, i\leq k)$ and the martingale difference $M(F_\Lambda| \F_k)- M( F_\Lambda| \F_{k-1})$.
This yields the lower bound
\begin{equation}
\label{eq:martdiff}
\text{Var}_M ( F_\Lambda)
\geq \sum_{k=1}^N \text{Var}_M\Bigl(M( F_\Lambda| \F_k) - M( F_\Lambda| \F_{k-1})\Bigr)\, .
\end{equation}
To prove this inequality we use~\eqref{eq: var}, in addition dropping the couplings between boxes.
Again using~\eqref{eq: var}, this time averaging the couplings {\it outside\/} a given $B_k$, gives
\begin{equation}
\text{Var}_M ( F_\Lambda)\geq \sum_{k=1}^N \text{Var}_M\Bigl(M( F_\Lambda| J_{B_k})\Bigr)\, .
\end{equation}

The remainder of the proof focuses on establishing the following two claims:
\begin{enumerate}
\item $\text{Var}_M\Bigl(M( F_\Lambda| J_{B_k})\Bigr)=\text{Var}_M\Bigl(M(F_\Lambda| J_{B_1})\Bigr)$, so the variance does not depend on $k$;
\item $\text{Var}_M\Bigl(M( F_\Lambda| J_{B_1})\Bigr)>c'$ for $c'>0$ independent of $\Lambda$.
\end{enumerate}
Once both claims are established, it follows that $\text{Var}_M ( F_\Lambda)\geq c' N= c' C|\Lambda|$, thereby concluding the proof.

\medskip

To prove the first claim, let $B$ be a generic block. Writing the variance using \eqref{eq: sym var} gives
\begin{equation}
\text{Var}_M\Bigl( M(F_\Lambda | J_B)\Bigr)
= \frac{1}{2}\int \nu(z''_B)\int \nu(z'_B) \Bigl\{M(F_\Lambda | J_B=z_B')-M(F_\Lambda | J_B=z_B'')\Bigr\}^2
\end{equation}
By Lemma \ref{lem: deriv}, $M(F_\Lambda| J_B)$ is differentiable with respect to $J_{xy}$ for any edge $(x,y)$ in $B$.  
In particular, the gradient $\nabla_B M(F_\Lambda | J_B)$ exists almost surely and we can write 
\begin{equation}
M(F_\Lambda | J_B=z_B')-M(F_\Lambda | J_B=z_B'')=\int_{z_B'\to z_B''} \nabla_B M(F_\Lambda | J_B=z_B)\cdot dz_B\ \text{ $\nu$-a.s.}
\end{equation}

For an edge $(x,y)$, define the function
\begin{equation}
\delta_{xy}(\Gamma, \Gamma')=\Gamma(\sigma_x\sigma_y)-\Gamma'(\sigma'_x\sigma'_y)\ .
\end{equation}
By Lemma~\ref{lem: deriv}, we have
\begin{equation}
\frac{\partial}{\partial J_{xy}} M(F_\Lambda | J_B)= \beta  M_{B^c}\bigl(\delta_{xy}(L_{J_B}\Gamma, L_{J_B}\Gamma')\bigr)\ .
\end{equation}
Lemma~\ref{lem: cond} applied to the function $\delta_{xy}$ implies that 
\begin{equation}
M_{B^c}\bigl(\delta_{xy}(L_{J_B}\Gamma, L_{J_B}\Gamma')\bigr)= M(\delta_{xy}(\Gamma, \Gamma')| J_B)\ \text{ $\nu$-a.s.}
\end{equation}
Writing $M\bigl(\delta_{B}(\Gamma, \Gamma')| J_B\bigr):= \bigl(M(\delta_{xy}(\Gamma, \Gamma')| J_B\bigr), (x,y)\in E(B))$, we have
\begin{equation}
\text{Var}_M\Bigl(M(F_\Lambda | J_B)\Bigr)
= \frac{\beta^2}{2}\int \nu(z''_B)\int \nu(z'_B) \Bigl\{  \int_{z_B'\to z_B''} M(\delta_{B}(\Gamma, \Gamma')| J_B=z_B)\cdot dz_B \Bigr\}^2
\end{equation}
For a translation $T$ such that $TB \subset \Lambda$, consider $\text{Var}_M\bigl( M(F_\Lambda | J_{TB})\bigr)$.  Note that
\begin{equation}
\delta_{TxTy}(\Gamma, \Gamma')=\Gamma(\sigma_{Tx}\sigma_{Ty})-\Gamma'(\sigma'_{Tx}\sigma'_{Ty})= 
T\Gamma(\sigma_x\sigma_y)-T\Gamma'(\sigma'_x\sigma'_y)\ .
\end{equation}
In particular, $\delta_{TxTy}(\Gamma, \Gamma')$ has the same law under $M$ as $\delta_{xy}(\Gamma, \Gamma')$ by Lemma~\ref{lem: M invariance}.
We conclude that
\begin{equation}
\text{Var}_M\Bigl(M(F_\Lambda | J_{TB})\Bigr)=\text{Var}_M\Bigl( M(F_\Lambda | J_B)\Bigr)\ .
\end{equation}
This proves Claim~1. 

\medskip

To prove Claim~2 we observe that, since the gradient does not depend on $\Lambda$, we simply need to show that
\begin{equation}
\int \nu(z''_B)\int \nu(z'_B) \Bigl\{  \int_{z_B'\to z_B''} M(\delta_{B}(\Gamma, \Gamma')| J_B=z_B)\cdot dz_B \Bigr\}^2>0\ .
\end{equation}

Suppose the contrary. Then 
\begin{equation}
\int_{z_B'\to z_B''} M(\delta_{B}(\Gamma, \Gamma')| J_B=z_B)\cdot dz_B=0 \ \ \text{ for $\nu$-almost all $(z_B',z_B'')$. }
\end{equation}
In particular, since $\nu$ is continuous, this implies
\begin{equation}
M\Bigl(\delta_{B}(\Gamma, \Gamma')| J_B\Bigr)=0 \ \text{ $\nu$-a.s.}
\end{equation}
In other words, for any edge $(x,y)\in E(B)$,
\begin{equation}
M\Bigl(\Gamma(\sigma_x\sigma_y)- \Gamma'(\sigma_x\sigma_y)  | J_B\Bigr) =0 \ \text{ $\nu$-a.s.}
\end{equation}
which contradicts~(\ref{eq: key}).

\end{proof}

\section{Proof of the results in $d=2$}
\label{sec: d=2}

The proof of Corollary~\ref{cor: 2d} proceeds by contradiction, studying the fluctuations of $M( F_\Lambda| J_\Lambda)$ 
in the spirit of~\cite{AW90}.

The outline of the proof is as follows.
We always have by Proposition \ref{prop: upper} below that
\begin{equation}
|M( F_\Lambda| J_\Lambda)|\leq 4\beta \sum_{e\in \partial \Lambda} \nu(|J_e|)=4\beta |\partial \Lambda| \nu(|J_e|)\ .
\end{equation}
In particular, for all $t>0$
\begin{equation}
\label{eq: upper thm}
\nu\Bigl(\exp t\frac{M( F_\Lambda| J_\Lambda)}{|\partial \Lambda|} \Bigr) \leq e^{4\beta t} \ .
\end{equation}
To get a contradiction, we would like to show that, under Assumption~\ref{ass}, the distribution of 
\begin{equation}
\frac{M( F_\Lambda| J_\Lambda)}{|\partial \Lambda|}
\end{equation}
has at least a Gaussian tail, in the sense that there exists $c>0$ such that 
\begin{equation}
\label{eq: lower thm}
\nu\Bigl(\exp t\frac{M( F_\Lambda| J_\Lambda)}{|\partial \Lambda|}\Bigr) \geq e^{c t^2}\ .
\end{equation}
This is possible in $d=2$ where, unlike in higher dimensions,
$|\partial \Lambda|$ is of the order of the square root of
$|E(\Lambda)|$.  This is not true in higher dimensions.
Eqs.~\eqref{eq: upper thm} and~\eqref{eq: lower thm} are in obvious
contradiction for $t$ large enough.

\subsection{Upper bound for the difference of free energies}
\begin{prop}
\label{prop: upper}
For $F_\Lambda$ defined in \eqref{eq: F}, we have
\begin{equation}
\label{eq: upper}
\Bigl| F_\Lambda(J,\Gamma,\Gamma')\Bigr| \leq 4\beta \sum_{e\in \partial \Lambda} |J_e|\ , \text{ $M$-a.s.}
\end{equation}
In particular, 
\begin{equation}
\limsup_{\Lambda\to \Z^d}  \frac{1}{|\partial \Lambda|} \Bigl|F_\Lambda(J,\Gamma,\Gamma')\Bigr| \leq 4\beta \ \nu(|J_e|) \qquad \text{ $M$-a.s.}
\end{equation}
\end{prop}

Write $G_{\beta, \Lambda, J}$ for the Gibbs measure in a finite box $\Lambda$. (Recall that the Hamiltonian $H_{\Lambda,J}$ does not include interaction terms between $\Lambda$ and the boundary $\partial \Lambda$). 

\begin{equation}
G_{\beta, \Lambda, J}\big(f(\sigma_\Lambda)\big)= 
 \frac{\sum_{\sigma_\Lambda \in \{-1,+1\}^\Lambda} f(\sigma_\Lambda) \exp(-\beta H_{\Lambda,J}(\sigma))}{\sum_{\sigma_\Lambda \in \{-1,+1\}^\Lambda} \exp(-\beta H_{\Lambda,J}(\sigma))}
\end{equation}
To prove the Proposition, we need the following lemma.
\begin{lem}
\label{lem: upper}
Let $\Gamma\in \mathcal G_{\beta,J}$ be a Gibbs state and $\Lambda$ a finite box of $\Z^d$.
For $f: \{-1,1\}^\Lambda\to [0,\infty)$ a positive measurable function, we have
\begin{equation}
\exp\Bigl(-2\beta \sum_{e\in \partial \Lambda} |J_e|\Bigr)\leq
\frac{\Gamma(f(\sigma_\Lambda))}{G_{\beta, \Lambda, J}\big(f(\sigma_\Lambda)\big)}
\leq 
\exp\Bigl(2\beta \sum_{e\in \partial \Lambda} |J_e|\Bigr)
\end{equation}
\end{lem}

\begin{proof}
Consider the Hamiltonian of interactions between spins in $\lambda$ and spins in $\Lambda^c$:
\begin{equation}
H_{\partial \Lambda, J}(\sigma)= \sum_{(x,y)\in \partial \Lambda}- J_{xy}\sigma_x\sigma_y\ .
\end{equation}
Clearly,
\begin{equation}
\label{eq: boundary bound}
\exp\Bigl(-\beta \sum_{e\in \partial \Lambda} |J_e|\Bigr)\leq \exp(-\beta H_{\partial\Lambda,J}(\sigma)) \leq \exp\Bigl(\beta \sum_{e\in \partial \Lambda} |J_e|\Bigr)
\end{equation}
Since $\Gamma$ is a Gibbs state the DLR equations imply
\begin{equation}
\Gamma(f(\sigma_\Lambda))
= \Gamma\Bigl(
\frac{\sum_{\sigma_\Lambda \in \{-1,+1\}^\Lambda} f(\sigma_\Lambda) \exp(-\beta (H_{\Lambda,J}(\sigma)) + H_{\partial\Lambda,J}(\sigma)) )}{\sum_{\sigma_\Lambda \in \{-1,+1\}^\Lambda} \exp(-\beta  (H_{\Lambda,J}(\sigma)) + H_{\partial\Lambda,J}(\sigma)))}
\Bigr)\ .
\end{equation}
Since $f$ is assumed positive, the bounds \eqref{eq: boundary bound} give the desired result.
\end{proof}

\begin{proof}[Proof of Proposition \ref{prop: upper}]
A direct application of Lemma \ref{lem: upper} with $f(\sigma_\Lambda)=\exp \beta H_{\Lambda,J}(\sigma)$ produces the bounds
\begin{equation}
-4\beta\sum_{e\in \partial \Lambda} |J_e|\leq \log \frac{\Gamma\big(\exp \beta H_{\Lambda,J}(\sigma) \big)}{\Gamma'\big(\exp \beta H_{\Lambda,J}(\sigma)\big)} \leq 4\beta \sum_{e\in \partial \Lambda} |J_e|\ .
\end{equation}
Therefore
\begin{equation}
\Bigl| \frac{F_\Lambda(J,\Gamma,\Gamma')}{|\partial \Lambda|}\Bigr| \leq \frac{4\beta}{|\partial \Lambda|}\sum_{e\in \partial \Lambda} |J_e|\ .
\end{equation}
The conclusion follows from the strong law of large numbers.

\end{proof}

\subsection{Proof of Corollary \ref{cor: 2d} }

To prove \eqref{eq: lower thm}, we need to improve Theorem \ref{thm: lower} from a lower bound on the variance to a lower bound on the moment generating function.
 For this, we use the following proposition which is a version of Proposition A.2.1 in~\cite{AW90}. 
It is a weaker result than the central limit theorem for martingale arrays.
It provides a lower bound on the moment generating function
of the martingale whenever a lower bound on the quadratic variation (in the probabilistic sense) is available.
\begin{prop}[Proposition A.2.1 in~\cite{AW90}]
\label{prop: clt}
For every $n\in \N$, let $(X_{n,k}, k\leq n)$ be a martingale with respect to the the filtration $(\F_{n,k}, k\leq n)$.
Suppose that the martingale differences $\Delta X_{n,k}:= X_{n,k}-X_{n,k-1}$ satisfy the two conditions below
\begin{enumerate}
\item There exists $\sigma^2$ such that 
\begin{equation}
\lim_{n\to\infty} \prob\Bigl( \sum_{k=1}^n \E[\Delta X_{n,k}^2 | \F_{n,k-1}] \leq \sigma^2-\delta\Bigr)\to 0 \ \text{ for every $\delta>0$}\ .
\end{equation}
\item For every $\delta >0$, 
\begin{equation}
\lim_{n\to\infty} \sum_{k=1}^n \E\Bigl[ \Delta X_{n,k}^2 1_{\{| \Delta X_{n,k}|>\delta\}}\Bigr] =0\ .
\end{equation}
Then
\begin{equation}
\liminf_{n\to\infty} \E[\exp t X_{n,n}]\geq e^{\frac{t^2\sigma^2}{2}}\ .
\end{equation}
\end{enumerate}
\end{prop}
The proposition will be used with the following quantities.
Consider the lexicographic order on $\Z^2$: $(x_1, x_2) \preceq (y_1, y_2)$ iff $x_1 < y_1$ or $x_1 = y_1$ and $x_2 \leq y_2$.  It induces an order on $E(\Z^2)$ by enumerating the edges in order of the vertices at which they originate, with (say) the edge on top of a vertex preceding the edge on its right. We will use the same notation for both orders and thus write $e\preceq e'$ if $e$ precedes $e'$ in the above sense on $E(\Z^2)$.

Define the following $\sigma$-algebras
\begin{equation}
\begin{aligned}
\F_{\preceq e}&:=\sigma\Bigl( J_{xy}: (x,y)\preceq e\Bigr)\\
\F_{\Lambda}&:=\sigma\Bigl( J_{xy}: (x,y)\in E(\Lambda) \Bigr)\\
\F_{\Lambda, \preceq e}&:= \F_{\Lambda}\cap \F_{\preceq e}
\end{aligned}
\end{equation}
It will be convenient to enumerate the edges in $E(\Lambda)$ in the above order: $E(\Lambda)=\{e_1,e_2,\dots, e_{k},\dots e_{|E(\Lambda)|}\}$.
We define
\begin{equation}
\begin{aligned}
Y_{\Lambda,k}&:= M\Bigl( F_\Lambda | \F_{\Lambda, \preceq e_k}\Bigr)\\
\Delta Y_{\Lambda,k}&:= Y_{\Lambda,k}-Y_{\Lambda,k-1}\ .
\end{aligned}
\end{equation}
Plainly, for each $\Lambda$, $(Y_{\Lambda,k}, k\leq |E(\Lambda)|)$ is a martingale for the filtration $(  \F_{\Lambda, \preceq e_k}, k\leq |E(\Lambda)|)$.
Note that with this notation $\sum_{k=1}^{|E(\Lambda)|}= \Delta Y_{\Lambda,k}=  M\Bigl( F_\Lambda | \F_{\Lambda}\Bigr)- M( F_\Lambda)$.
Equation \eqref{eq: lower thm} will be proved by showing that 
\begin{equation}\Bigl(\frac{Y_{\Lambda,k}}{\sqrt{|E(\Lambda)|}}, k\leq |E(\Lambda)|\Bigr)\end{equation}
 satisfies the two hypotheses of Proposition \ref{prop: clt}.
In what follows it will be convenient to use the notation $J_{\prec e}$ to denote $(J_{e'}, e'\prec e)$ and similarly for $J_{\succ e}$. 
We can define using this notation
\begin{equation}
\overline F_\Lambda(J)= \overline F_\Lambda(J_{\prec e}, J_e, J_{\succ e}):= M(F_\Lambda| \F_\Lambda)\ .
\end{equation}

We start with the simpler part.
\begin{proof}[Proof of Hypothesis 2)]
Observe that, by Lemma \ref{lem: deriv},  $\Delta Y_{\Lambda,k}$ can be represented as
\begin{equation}
\label{eq: delta deriv}
\begin{aligned}
\Delta Y_{\Lambda,k}(J_{e_k}, J_{\prec e_k})&=\int \nu(d J_{\Lambda^c})\int \nu(dJ_{\succ e_k} )\Bigl\{  \overline F_\Lambda(J_{\prec e_k}, J_{e_k}, J_{\succ e_k}) -\int \nu(dy)\overline F_\Lambda(J_{\prec e_k}, y, J_{\succ e_k}) \Bigr\}\\
&= \int \nu(d J_{\Lambda^c})\int \nu(dJ_{\succ e_k} )\int \nu(dy)\Bigl\{\int_y^{J_{e_k}}D_k(J_{\prec e_k}, s, J_{\succ e_k})\ ds \Bigr\}
\end{aligned}
\end{equation}
where for $e_k=(x,y)$ we defined $D_k(J)=\beta\Bigl\{ \kappa_J(\Gamma(\sigma_x\sigma_y))-\kappa'_J(\Gamma'(\sigma_x\sigma_y))\Bigr\}$.
In particular,
\begin{equation}
|\Delta Y_{\Lambda,k}(J_{e_k}, J_{\prec e_k})| \leq 2\beta \int \nu(dy) |J_{e_k}-y|\leq 2\beta( |J_{e_k}| +\nu(|J_e|))\ ,
\end{equation}
and 
\begin{equation}
\nu\Bigl(\Delta Y_{\Lambda,k}^2\Bigr)\leq 16\beta^2\  \nu(J_e^2)\ ,
\end{equation}
for all $k\leq |E(\Lambda)|$. Also $\nu\Bigl(\Delta Y_{\Lambda,k}^4\Bigr)<\infty$ uniformly in $k$ if we suppose that $\nu(J_e^4)<\infty$. 

For $\delta>0$, the above implies that
\begin{equation}
\nu\Bigl\{ |\Delta Y_{\Lambda,k}|>\delta \sqrt{|E(\Lambda)|}\Bigr\}\leq \frac{\nu\Bigl(\Delta Y_{\Lambda,k}^2\Bigr)}{\delta^2|E(\Lambda)|} \to 0\  \text{as $\Lambda\to \Z^d$, uniformly in $k$.}
\end{equation}
Hence, by Cauchy's inequality, for some constant $C>0$
\begin{equation}
 \sum_{k=1}^{|E(\Lambda)|}\nu \Bigl(\frac{\Delta Y_{\Lambda,k}^2}{|E(\Lambda)|} 1_{\{| \Delta Y_{\Lambda,k}|>\delta|E(\Lambda)|^{1/2}\}} \Bigr)  \leq  C  \frac{\max_k\ \nu\Bigl(\Delta Y_{\Lambda,k}^4\Bigr)^{1/2}\nu(J_e^2)^{1/2}}{\delta|E(\Lambda)|^{1/2}}\to 0  \ .
\end{equation}
\end{proof}

\begin{proof}[Proof of Hypothesis 1)]
We show that
\begin{equation}
\label{eqn: hypo2}
\lim_{\Lambda\to\Z^d}\frac{1}{|E(\Lambda)|}  \sum_{k=1}^{|E(\Lambda)|}\nu(\Delta Y^2_{\Lambda,k}| \F_{\Lambda, \preceq e_{k-1}})= \sigma^2 
\end{equation}
where the convergence holds in $L^1(\nu)$ (hence also  in $\nu$-probability) and $\sigma^2$ is a constant.
Note that, under Assumption \ref{ass}, $\sigma^2$ is non-zero.
 Indeed, we have by conditioning
\begin{equation}
\frac{1}{|E(\Lambda)|}\text{Var}\Bigl( M( F_\Lambda| \F_\Lambda) \Bigr)=\frac{1}{|E(\Lambda)|}\text{Var}\Bigl( \sum_{k=1}^{|E(\Lambda)|} \Delta Y_{\Lambda,k}\Bigr)
=\frac{1}{|E(\Lambda)|}\sum_{k=1}^{|E(\Lambda)|} \nu\Bigl(\Delta Y^2_{\Lambda,k}\Bigr)\ .
\end{equation}
The right side converges to $\sigma^2$ by taking the expectation in \eqref{eqn: hypo2} ($L^1$-convergence implies convergence of the expectations). 
On the other hand, the proof of Theorem \ref{thm: lower} shows that the left side is strictly greater than $0$ uniformly in $\Lambda$.

The proof is based on the $L^2$-ergodic theorem. To use it, we need to work around the dependence on $\Lambda$.
For this purpose, define the following quantities, 
\begin{equation}
\begin{aligned}
\widetilde Y_{\Lambda,k}&:=M\Bigl( F_\Lambda | \F_{\preceq e_k}\Bigr)\\
\Delta \widetilde Y_k&:=Y_k-Y_{k-1}
\end{aligned}
\end{equation}
Note that $Y_{\Lambda, k}=\nu(\widetilde Y_{\Lambda,k}|\F_\Lambda).$ 
Similarly as in \eqref{eq: delta deriv}
\begin{equation}
\label{eq: delta deriv2}
\begin{aligned}
\Delta \widetilde Y_k 
&= \int \nu(dJ_{\succ e_k}) \int \nu(dy)\Bigl\{\int_y^{J_{e_k}}D_k(J_{\prec e_k}, s, J_{\succ e_k})\ ds \Bigr\}
\end{aligned}
\end{equation}
where for $e_k=(x,y)$ we defined $D_k(J)=\beta\Bigl\{ \kappa_J(\Gamma(\sigma_x\sigma_y))-\kappa'_J(\Gamma'(\sigma_x\sigma_y))\Bigr\}$.
In particular, $\Delta \widetilde Y_k$ does not depend on $\Lambda$ (as the notation suggests). Moreover, it is easily checked that $\nu(\Delta \widetilde Y_k ^4)<\infty)$.

The crucial point is the fact that the random variables $f_{e_k}(J):=\nu(\Delta \widetilde Y_k^2| \F_{\preceq e_{k-1}})$ are translates, that is if $T$ is a translation mapping $e_k$ to $e_{l}$ then
\begin{equation}
f_{e_{l}}(J)=f_{e_k}(T^{-1}J)\ .
\end{equation}
(Note that horizontal and vertical edges are not translates. Thus, the
functions $f_{e_k}(J)$ can be reduced to two functions up to
translation, not one.)  We can therefore apply von Neumann's ergodic
theorem (see e.g. \cite{reed-simon}) to conclude that
\begin{equation}
\label{eqn: ergodic}
\frac{1}{|E(\Lambda)|}  \sum_{k=1}^{|E(\Lambda)|} \nu(\Delta \widetilde Y_{k}^2|  \F_{\preceq e_{k-1}})
\end{equation}
converges in $L^2(\nu)$ to a constant, say $\sigma^2$.

To prove \eqref{eqn: hypo2}, it remains to show that we can replace  $\nu(\Delta \widetilde Y_{k}^2|  \F_{\preceq e_{k-1}})$ by $\nu(\Delta \widetilde Y_{k}^2|  \F_{\Lambda,\preceq e_{k-1}})$ and
 $\nu(\Delta \widetilde Y_{k}^2|  \F_{\Lambda,\preceq e_{k-1}})$ by  $\nu(\Delta Y_{\Lambda, k}^2|  \F_{\Lambda,\preceq e_{k-1}})$
in \eqref{eqn: ergodic}. In other words, we need to show that we do not lose much by averaging the couplings outside $\Lambda$. 
Precisely, we show, as $\Lambda\to \Z^2$
\begin{eqnarray}
\label{eq: approx1}
\nu\Bigl(\Bigl| \nu(\Delta \widetilde Y_{k}^2|  \F_{\Lambda,\preceq e_{k-1}}) -  \nu(\Delta Y_{\Lambda, k}^2|  \F_{\Lambda,\preceq e_{k-1}})\Bigr|\Bigr) &\to 0 \\
\label{eq: approx2}
\nu\Bigl(\Bigl| \nu(\Delta \widetilde Y_{k}^2|  \F_{\preceq e_{k-1}}) -  \nu(\Delta \widetilde Y_{\Lambda, k}^2|  \F_{\Lambda,\preceq e_{k-1}})\Bigr|\Bigr) &\to 0
\end{eqnarray}
uniformly in $k$. 
(In fact, as it will be clear from our reasoning, the convergence is uniform for $e_k$ not too close to the boundary of $\Lambda$.)
The convergence  is in $L^1$ and enough for our purpose.

We first show \eqref{eq: approx1}. By Jensen's inequality, the left-hand side is smaller than
\begin{equation}
\nu\Bigl(\Bigl| \Delta \widetilde Y_{k}^2-  \Delta Y_{\Lambda, k}^2\Bigr|\Bigr)
\end{equation}
Factoring the difference of the squares and using the Cauchy-Schwarz inequality, we get the upper bound 
\begin{equation}
\nu\Bigl(\Bigl| \Delta \widetilde Y_{k}-  \Delta Y_{\Lambda, k}\Bigr|^2\Bigr)^{1/2}\nu\Bigl(\Bigl| \Delta \widetilde Y_{k}+ \Delta Y_{\Lambda, k}\Bigr|^2\Bigr)^{1/2}
\end{equation}
The second term is finite, uniformly in $k$, as a consequence of triangle's inequality and equations \eqref{eq: delta deriv} and \eqref{eq: delta deriv2}.
It remains to prove that $\nu\Bigl(\Bigl| \Delta \widetilde Y_{k}-  \Delta Y_{\Lambda, k}\Bigr|^2\Bigr)$ goes to $0$ uniformly in $k$.
Again, by Jensen's inequality and the representations \eqref{eq: delta deriv} and \eqref{eq: delta deriv2}, we have
\begin{equation}
\nu\Bigl(\Bigl| \Delta \widetilde Y_{k}-  \Delta Y_{\Lambda, k}\Bigr|^2\Bigr)\leq
\nu\Bigl(\Bigl|W_k- \nu(W_k|\F_\Lambda)\Bigr|^2\Bigr)
\end{equation}
where 
\begin{equation}
W_k=W_k(J_{\prec e_k}, J_{\succ e_k}):=\int \nu(dy)\Bigl\{\int_y^{J_{e_k}}D_k(J_{\prec e_k}, s, J_{\succ e_k})\ ds \Bigr\}\ .
\end{equation}
Note that all $W_k$'s have the same $L^2$-norm since the function is
translation-invariant. However, the $e_k$'s do not have the same
location in the box $\Lambda$. Thus, we cannot conclude right away,
using the $L^2$-martingale convergence theorem, that the term converges
to $0$ uniformly in $k$. To get around this, note that
$\nu\Bigl(\Bigl|W_k- \nu(W_k|\F_\Lambda)\Bigr|^2\Bigr)$ decreases as
$\Lambda$ grows. Therefore, it suffices to take the conditioning
on a smaller box than $\Lambda$. For example, take each $e_k$ that is
at least $|\Lambda|^{1/4}$ (say) away from the boundary (there are
$(1+o(1))|E(\Lambda)|$ such edges), and consider a box $B_k$ of area
$|\Lambda|^{1/2}$ centered at $e_k$. Then, by the previous remark,
\begin{equation}
\nu\Bigl(\Bigl|W_k- \nu(W_k|\F_\Lambda)\Bigr|^2\Bigr)\leq \nu\Bigl(\Bigl|W_k- \nu(W_k|\F_{B_k})\Bigr|^2\Bigr)
\end{equation}
where $\F_{B_k}$ is a box of area $|\Lambda|^{1/2}$ centered at $e_k$. The right side is now independent on $k$ and goes to $0$ as $\Lambda\to \Z^2$ by  the $L^2$-martingale convergence theorem.

For \eqref{eq: approx2}, the term is smaller than
\begin{equation}
\nu\Bigl(\Bigl| \Delta \widetilde Y_{k}^2 -  \nu(\Delta \widetilde Y_{\Lambda, k}^2|  \F_{\Lambda})\Bigr|\Bigr)\ .
\end{equation}
With the notation introduced above, we can write
\begin{equation}
 \Delta \widetilde Y_{k}= \nu( W_k| \F_{\preceq e_k}). 
 \end{equation}
 Again, this quantity does not depend on $k$. 
 Using the same reasoning as above by reducing $\F_\Lambda$ to $\F_{B_k}$, we deduce the convergence to $0$ uniformly in $k$ for $e_k$ at least $|\Lambda|^{1/4}$ away from the boundary.
\end{proof}

\section{Proof of Corollary \ref{cor: countable}}

The idea of the proof is to use the weights in the decomposition of
the metastate as a {\it tag} to distinguish the states as $J$ is
varied.

Suppose that the metastate $\kappa_\cdot$ is of the form $\sum_{\alpha \in \mathcal A}p_\alpha \delta_{\Gamma^\alpha}$. 
First, we observe that the set of weight $(p_\alpha,\alpha \in \mathcal A)$ does not depend on $J$.  
Indeed, by the translation covariance of the metastate $\kappa_{TJ}(d\Gamma)=\kappa_J(d\ T\Gamma)$. 
In particular, the weight associated with $\Gamma_\alpha$ is the same as the weight of $T\Gamma_\alpha$. 
This implies that the set of weights is translation-invariant as a function of $J$. Thus, it must be constant $\nu$-almost surely.
By hypothesis, there must exist $\alpha,\alpha'$ such that $p_\alpha\neq p_\alpha'$. Let $\Gamma^\alpha$ and $\Gamma^{\alpha'}$ 
be the corresponding states. These states belong to $\mathcal G_J$ and we write $\Gamma^\alpha=\Gamma_J^\alpha$ and $\Gamma^{\alpha'}=\Gamma_J^{\alpha'}$
to make the dependence on $J$ explicit. 
Therefore each tag yields a well-defined measurable map of the couplings to $M_1(\Sigma)$
\begin{equation}
J\mapsto \kappa^\alpha_J:=\delta_{\Gamma^\alpha_J} \qquad J\mapsto \kappa^{\alpha'}_J:=\delta_{\Gamma^{\alpha'}_J} \ .
\end{equation}
(If more than one state has weight $p_\alpha$, we can take $\Gamma^\alpha$ 
to be the weighted linear combination of these states without loss of generality.)

To prove the corollary, we show that $\kappa^\alpha_J$ and
$\kappa^{\alpha'}_J$ satisfy the properties of a metastate in
Definition \ref{df: metastate} and also satisfies Assumption
\ref{ass}.  This is in contradiction with Corollary \ref{cor: 2d}.
The property of support of metastates on Gibbs states is obvious.  The
translation covariance is also clear since the weight $p_\alpha$ puts
$\Gamma^\alpha_J$ and $T\Gamma^\alpha_J$ in bijection by the
translation covariance of the original metastate $\kappa$. For the
same reason, $\Gamma^\alpha_J$ is mapped to $L_{J_B}\Gamma^\alpha_J$
when we consider $\kappa_{J+J_B}$. This ensures coupling
covariance. Therefore, Corollary \ref{cor: 2d} holds.  On the other
hand, by hypothesis, $\Gamma^\alpha$ and $\Gamma^{\alpha'}$ are
incongruent, thus by Definition \ref{df: incongruent}, there exists
$\varepsilon>0$ such that
\begin{equation}
\label{eqn: density}
\liminf_{\Lambda\to\Z^2}\frac{1}{|E(\Lambda)|} \sum_{(x,y)\in E(\Lambda)} 1_{\{(x,y)\in E(\Lambda): |\Gamma^\alpha_J(\sigma_x\sigma_y)-\Gamma^{\alpha'}_J(\sigma_x\sigma_y)|>\varepsilon\}} >0
\end{equation}
The left-hand side is a translation-invariant function of the $J$ since $\Gamma^\alpha_{TJ}=T\Gamma_J^\alpha$. 
In particular, the limit  in \eqref{eqn: density} exists and equals
\begin{equation}
\nu\left\{J: |\Gamma^\alpha_J(\sigma_x\sigma_y)-\Gamma^{\alpha'}_J(\sigma_x\sigma_y)|>\varepsilon \right\}
\end{equation}
In particular, Assumption \ref{ass} is fulfilled. This concludes the proof of the corollary.

\end{document}